\documentclass[11pt,a4paper]{amsart}
\usepackage[margin=2.5cm]{geometry}
\usepackage[T1]{fontenc}

\usepackage[utf8]{inputenc}
\usepackage{amssymb, amsthm, amsopn}
\usepackage[foot]{amsaddr} 
\usepackage{dsfont}
\usepackage{graphicx}

\usepackage{wrapfig}

\usepackage{subcaption}
\usepackage{booktabs}
\usepackage[section]{placeins}
\usepackage[shortlabels]{enumitem}
\usepackage{mathtools}
\usepackage{upgreek}
\usepackage{xcolor}
\usepackage{algorithm}
\usepackage{algorithmic}

\usepackage{thm-restate}
\usepackage{makecell}

\usepackage[breaklinks]{hyperref}
\hypersetup{
    colorlinks,
    linkcolor={red!40!black},
    citecolor={blue!50!black},
    urlcolor={blue!20!black}
}
\usepackage{doi}

\DeclareMathOperator{\tr}{tr}

\newcommand{\ot}[0]{\otimes}
\newcommand{\one}[0]{\mathds{1}}
\renewcommand{\a}{\alpha}

\newcommand{\C}{\mathds{C}}

\newcommand{\Q}{\mathbb{Q}}

\newcommand{\PP}{\mathcal{P}}

\newtheorem{theorem}{Theorem}
\newtheorem*{theorem*}{Theorem}
\newtheorem{proposition}[theorem]{Proposition}
\newtheorem{lemma}[theorem]{Lemma}

\newtheorem{corollary}[theorem]{Corollary}
\newtheorem{observation}[theorem]{Observation}
\newtheorem{example}[theorem]{Example}

\newtheorem*{problem*}{Problem}
\newtheorem*{question*}{Question}

\newtheorem*{result*}{Result}

\newcommand{\nn}{\nonumber}

\graphicspath{{code/figs/}}

\begin{document}
\title
[Quasi-Clifford to qubit mappings]
{Quasi-Clifford to qubit mappings}
\date{\today}

\author{Felix Huber}

\address{
Division of Quantum Computing,
Faculty of Mathematics, Physics and Informatics,
University of Gdańsk,
Wita Stwosza 57, 80-308 Gdańsk, Poland
}
\email{felix.huber@ug.edu.pl}

\thanks{
FH thanks
Claudio Procesi for careful explanations regarding Ref.~\cite{gastineau-hills_1982};
Paweł Cieśliński,
Tomáss Crosta,
and Michał Studziński
for discussions;
and
Alexssandre de Oliveira Junior,
Chau Nguyen,
Jędrzej Stempin,
and Nikolai Wyderka
for feedback on an early version of this article.
This research was funded in whole or in part by National Science Centre, Poland 2024/54/E/ST2/00451.
For the purpose of Open Access, the author has applied a CC-BY public copyright licence to any
Author Accepted Manuscript (AAM) version arising from this submission.
}

\begin{abstract}
Algebras with given (anti-)commutativity structure
are widespread in quantum mechanics. This structure is captured by quasi-Clifford algebras (QCA):
a QCA generated by $\alpha_1, \dots, \alpha_n$ is
is given by the relations  $\alpha_i^2 =  k_i$ and $\alpha_j \alpha_i = (-1)^{\chi_{ij}} \alpha_i \alpha_j$,
where $k_i \in \mathds{C}$ and $\chi_{ij} \in \{0, 1\}$.
We present a mapping from QCA to Pauli algebras
and discuss its use in quantum information and computation.
The mapping also provides a Wedderburn decomposition of matrix groups with quasi-Clifford structure.
This provides a block-diagonalization for e.g. Pauli groups,
while for Majorana operators the Jordan-Wigner transform is recovered.
Applications to the symmetry reduction of semidefinite programs and
for constructing maximal anti-commuting subsets are discussed.
\end{abstract}


\maketitle
\setcounter{tocdepth}{1}

Algebras with a given (anti-)commutativity structure
are widespread in quantum many-body physics and quantum computation.
Examples are found in fermion-to-qubit mappings for
electronic structure calculations~\cite{PRXQuantum.4.030314},
fermionic quantum computation~\cite{BRAVYI2002210},
the simulation of quantum circuits~\cite{PhysRevA.70.052328, Miller_2024},
the diagonalization of Hamiltonians~\cite{
PhysRevA.107.062416, PhysRevA.109.022618},
and
the construction of entanglement-assisted quantum codes~\cite{
PhysRevA.77.064302, PhysRevA.79.062322}.
In these applications,
one typically makes use of such algebras in their block-diagonalized form
to simplify calculations.
The aim of this paper is to provide a complementary viewpoint on these results
by linking them to the structural
decomposition of quasi-Clifford algebras (QCA).
In particular we want to answer the question: How can one realize a given anti-commutativity structure in terms of Pauli operators?

QCAs encode the (anti-) commutativity relations of a group or algebra, while ignoring any
possible additional structure~\cite{gastineau-hills_1982}.
A quasi-Clifford algebra generated by $\a_1, \dots, \a_n$
is given by the relations  $\a_i^2 = k_i$ and $\a_j \a_i = (-1)^{\chi_{ij}} \a_i \a_j$,
where $k_i \in \C$ and $\chi_{ij} = \chi_{ji} \in \{0, 1\}$~\footnote{
More generally, a QCA over a field $F$ can have $k_i \in F$.
We will mostly deal with {\it special}
quasi-Clifford algebras for which $\a_i^2 = k_i \in \{1, -1\}$~\cite{gastineau-hills_1982}.
}.
In applications, $\chi_{ij}$ is often given as the adjacency matrix of an anti-commutativity graph.
Examples of algebras with such structure
are subgroups of the $n$-qubit Pauli group and Majorana operators,
describing multi-qubit and multi-fermion systems respectively.
The quasi-Clifford structure of a given algebra keeps track of its (anti-) commutativity relations
while ignoring additional group or algebra structure.
That such approach can be useful is justified by the fact that (anti-) commutativity encodes
key features of many problems in quantum physics, e.g. for
ground state energies~\cite{
10.1145/3519935.3519960,
cbqf-d24r,
hastings2024limitationsseparationsquantumsumofsquares,
cbqf-d24r},
uncertainty relations~\cite{PhysRevA.107.062211, PhysRevLett.132.200202}, numerical ranges~\cite{PRXQuantum.5.020318},
and state tomography~\cite{PRXQuantum.6.010336}.

Gastineau-Hills has shown that over the field of complex numbers $\C$ every
QCA decomposes into a direct sum of tensor products of two-dimensional
representations~\cite{gastineau-hills_1982}.
We take this approach to construct quasi-Clifford to qubit mappings,
so to realize any given anti-commutativity structure on multi-qubit systems.
This shows that previously introduced Pauli algebra decompositions also apply to QCAs~\cite{PhysRevA.79.062322}.
The method also recovers the Jordan-Wigner transformation mapping Majorana operators to qubits,
providing a pleasing link to condensed matter physics.
Finally, we show applications to the symmetry-reduction of semidefinite programs
and for constructing maximal anti-commuting subsets of Pauli groups.

\begin{figure}[tbp]
 \includegraphics[height = 10em]{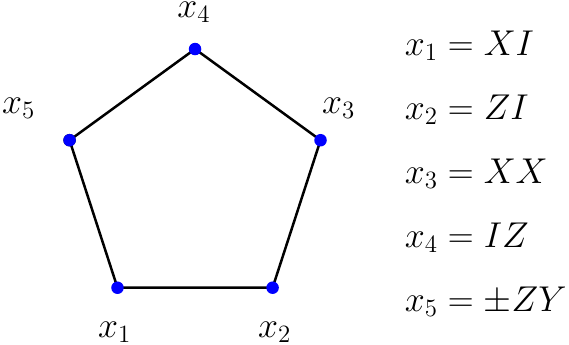}
 \caption{\label{fig:pentagon}
 A quasi-Clifford to qubit mapping that realizes the anti-commutation relations given by the pentagon:
 $x_i x_j = - x_j x_i$ if $i\sim j$ and
 $x_i x_j = x_j x_i$ if $i \not \sim j$.
 The two irreducible representations of this quasi-Clifford algebra are labeled by the sign of $x_5 = \pm ZY$.
 }
\end{figure}

\smallskip
\noindent{\bf Notation.}
Let $X_\ell, Y_\ell, Z_\ell$ be Pauli matrices that act on coordinate $\ell$ by $X, Y,Z$ respectively
and by the identity matrix $\one$ on the remaining coordinates.
$I$ is the $2\times 2$ identity matrix.
We write tensor products $I \otimes X\otimes Y \otimes Z$ as $IXYZ$,
termed {\em Pauli strings}.
The group generated by Pauli strings of length~$n$ is the Pauli group~$\PP_n$.
We write $g_i \sim g_j$ if two Pauli strings anti-commute,
that is if $\{g_i, g_j\} = g_ig_j +  g_jg_i = 0$.
If $g_i$ and $g_j$ commute then $g_i \not \sim g_j$.
The anti-commutativity or frustration graph of a set of generators $\{g_i\}_{i=1}^m$
has edges $i\sim j$ if $g_i \sim g_j$ anti-commute.
Finally we denote by $G = \langle g_1,  \dots, g_m\rangle$
the group generated by $g_1, \dots, g_m$
and by $A = [x_1, \dots, x_m]$ the algebra generated by $x_1, \dots, x_m$.

\smallskip
\noindent{\bf Related works and Outlook.}
Refs.~\cite{PhysRevA.70.052328, PhysRevA.79.062322} contain perhaps the core of these types of decompositions,
making use of the Heisenberg representation and transformations that preserve the
symplectic inner product.
Refs.~\cite{PhysRevA.107.062416, PhysRevA.109.022618}
perform a similar reduction by block-diagonalizing the commutation matrix while keeping the symplectic product invariant.
Ref.~\cite{Miller_2024} optimizes quantum circuits to diagonalize a given set of commuting Pauli strings and
Ref.~\cite{vandenBerg2020circuitoptimization} reduces circuit complexity based on maximal commuting subsets.
Ref.~\cite{aguilar2024classificationpauliliealgebras, Wiersema2024} decompose Pauli Lie groups,
also making use of the quasi-Clifford Lie theory~\cite{
Gintz2018,
khovanova2008Cliffordalgebrasgraphs,
cuypers2021quasiCliffordalgebrasquadraticforms}.
Refs.~\cite{Sarkar2024quditpauligroupnon, makuta2025frustrationgraphformalismqudit} provide decompositions for qudit Pauli groups.
Ref.~\cite{Leopardi2019} contains a good overview on the QC algebra literature.
Many quantum computers have limited connectivity,
so it is interesting to design quasi-Clifford to qubit mappings
that are resource-efficient, for example that yield geometrically local and low weight observables.
In spirit of Ref.~\cite{Miller_2024}, it is interesting to ask what are mappings that are
sufficiently local or whose corresponding observables satisfy a given hardware connectivity structure.

\section{Quasi-Clifford algebras}
A quasi-Clifford algebra of $m$ generators $\a_1, \dots, \a_m$ over a field $F$ is given by relations
\begin{align}
 \a_i^2 = k_i \,,
 \quad \quad \a_j \a_i = (-1)^{\chi_{ij}} \a_i \a_j\,,
 \quad \quad \chi_{ij}=\chi_{ji} \in \{0,1\}\,,
\end{align}
with $k_i \in F$.
In what follows we choose $F=\C$ the field of complex numbers and $k_i \in \C$.
Denote by $\mathbb{C}_{b}$ be the quasi-Clifford algebra generated by $\beta$
with $\beta^2 = b$.
Denote by $\mathbb{Q}_{c,d}$ be the quasi-Clifford algebra generated by $\gamma$ and $\delta$
with $\gamma^2= c$,
$\delta^2 = d$,
and $\gamma\delta = -\delta \gamma$.
Write $[x_1, \dots, x_m]$ for the algebra generated by $x_1, \dots, x_m$ over $\C$.
Then $\mathbb{C}_{b} = [\beta]$ and $\mathbb{Q}_{c,d} = [\gamma,\delta]$.
Finally, let $G$ be the anti-commutativity graph of $\a_1, \dots, \a_m$,
with $i$ connected to $j$ (written $i\sim j$) if $\a_j \a_i = - \a_i \a_j$ and $i \not \sim j$ otherwise.

Any finite-dimensional semi-simple algebra
decomposes into a direct sum of finitely many matrix rings,
this is known as the Wedderburn decomposition.
The following is known about the structure of QC algebras:
\begin{theorem}[Structure of quasi-Clifford Algebras, {~\cite[Theorem 2.7]{gastineau-hills_1982}}]
\label{thm:wedder}
The Wedderburn structure of a quasi-Clifford algebra is
\begin{align}\label{eq:QCA_wedder}
\mathcal{C} &\simeq \mathbb{C}_{b_1} \ot \dots \ot \mathbb{C}_{b_r} \ot \Q_{c_1, d_1} \ot \dots \ot \Q_{c_s, d_s} \nn\\
            &= [\beta_1] \ot \dots \ot [\beta_r] \ot [\gamma_1, \delta_1] \ot \dots \ot [\gamma_s, \delta_s]\,,
\end{align}
where $r,s \geq 0$, $r+2s = m$,
and each $b_i,c_i, d_i$ is plus minus the product of some $k_i$'s.
The center is $[\beta_1] \ot \dots \ot [\beta_r]$ and has dimension $2^r$.
\end{theorem}

\smallskip
Let us consider the case where all elements $\alpha_i$ square to plus or minus one, $\a_i^2 = \pm 1$.
For $F = \C$, one has the following irreducible representations
of the structural constituents appearing in Theorem~\ref{thm:wedder}~\cite{gastineau-hills_1982},
\begin{align}\label{eq:reps}
 \mathbb{C}_1 = [\beta]: \quad\quad     &\beta \to 1 \quad \text{or} \quad \beta  \to -1\,, \nn\\
 \Q_{\pm 1,1} = [\gamma,\delta] : \quad\quad &\gamma \to \begin{pmatrix}
                            0 & \pm 1 \\
                            1 & \phantom{-}0
                           \end{pmatrix}\,, \quad
                \delta \to \begin{pmatrix}
                            1 & \phantom{-}0 \\
                            0 & -1
                           \end{pmatrix}\,, \nn\\
 \Q_{1,-1} = [\gamma,\delta] : \quad\quad &\gamma \to \begin{pmatrix}
                            0 & \phantom{-}1 \\
                            1 & \phantom{-}0
                           \end{pmatrix}\,, \quad
                \delta \to \begin{pmatrix}
                            0 & -1 \\
                            1 & \phantom{-}0
                           \end{pmatrix}\,, \nn\\
 \Q_{-1,-1} = [\gamma,\delta] : \quad\quad &\gamma \to \begin{pmatrix}
                            0 & -1 \\
                            1 & \phantom{-}0
                           \end{pmatrix}\,, \,\quad
                \delta \to \begin{pmatrix}
                            i & \phantom{-}0 \\
                            0 & -i
                           \end{pmatrix}\,.
\end{align}

Note that over $F = \mathds{C}$, the generators of
$\Q_{\pm 1,1}$, $\Q_{1,-1}$, and $\Q_{-1,-1}$
span the same algebra.
Then $[\gamma, \delta]$ spans the matrix space of a single qubit,
for which $\{\one,X,Y,Z\}$ is a basis.
As a consequence of Theorem~\ref{thm:wedder},
a QCA generated by $m$ Pauli strings
can always be faithfully represented on a
$s$-qubit Pauli algebra tensored with $r$ classical binary systems such that $r+2s=m$.

\smallskip
A constructive decomposition realizing Theorem~\ref{thm:wedder} is the following.
\begin{observation}[Splitting algorithm, {\cite[page 7]{gastineau-hills_1982}}]
\label{prop:split}
Let $[x_1, \dots, x_m]$ form a quasi-Clifford algebra.
The following provides a decomposition of quasi-Clifford algebras in terms of the components
$[\beta_i]$ and
$[\gamma_i, \delta_i]$.
Suppose that $x_1, x_2$ anti-commute, and set $y_1 = x_1, y_2 = x_2$.
Then for $i=3,\dots, m$ set
\begin{align}\label{eq:split_step}
 y_i &= \begin{cases}
        x_i &     \text{if $x_i$ commutes with $x_1$ and $x_2$}\,, \\
        x_1 x_i & \text{if $x_i$ commutes with $x_1$, does not commute with $x_2$}\,, \\
        x_2 x_i & \text{if $x_i$ does not commute with $x_1$, commutes with $x_2$}\,, \\
        x_1 x_2 x_i & \text{if $x_i$ does not commute with $x_1$ and $x_2$}\,,
       \end{cases}
\end{align}
and $k_i = y_i^2$.
This decouples (splits) the pair of vertices $y_1 \sim y_2$,
so that they commute with all remaining ones.
It also induces new relations on the remaining vertices.
The resulting graph is
$G = (1,2) \cup G_{m-2}$, where
$G_{m-2}$ is determined by the (new) commutation relations among the $y_3, \dots, y_m$.
An iterated application of this transformation
to the remaining graph decomposes an arbitrary graph into
$r$ isolated vertices and $s$ pairs of vertices connected by an edge,
so that $r+2s = m$.
Then the center of the QCA is composed of the isolated vertices $\mathbb{C}_b = [\beta]$,
while each connected pair forms a $\Q_{c,d} = [\gamma, \delta]$-algebra.
\end{observation}

The relations after a splitting step are given in Appendix~\ref{app:rel_after_split}.
Importantly, note that after each step the $y_1, \dots, y_m$
generate the same algebra as the original $x_1, \dots, x_m$.

\begin{figure}[btp]
 \includegraphics[height = 10em]{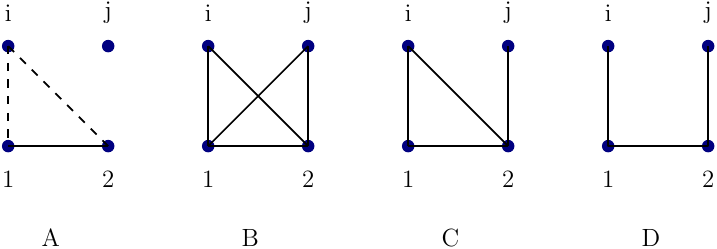}
 \caption{\label{fig:split_situations}
 Four types of edge configurations before the split. Dashed lines represent possible edges,
 while edges between $i,j$ are not drawn.\\
 {\bf A, B:} Splitting vertices $\{1,2\}$, the commutation relations between $i$ and $j$ remain.\\
 {\bf C, D:} The commutation relations between $i$ and $j$ invert.
 }
\end{figure}

\section{Quasi-Clifford to qubit mappings}

Let us now answer the question: given a (anti-)commutativity structure for a set operators,
how can one realize it in terms of Pauli strings?
We start with an example, using the algorithm in Observation~\ref{prop:split}.
\begin{example}[Pentagon] \label{eq:pentagon1}
Consider a quasi-Clifford algebra corresponding to the pentagon,
with edges given by $\{ \{1,2\}, \{2,3\}, \{3,4\}, \{4,5\}, \{1,5\}\}$.
Splitting the connected vertices $1$ and $2$ from the graph with the algorithm in Observation~\ref{prop:split} gives
\begin{align}
 y_1 = x_1 \,,\quad
 y_2 = x_2 \,,\quad
 y_3 = x_1 x_3 \,,\quad
 y_4 = x_4 \,,\quad
 y_5 = x_2 x_5\,,
\end{align}
with the anti-commutativity graph $\{1,2\}\cup \{\{3,4\}, \{4,5\}, \{3,5\} \}$ remaining.
Splitting the connected vertices $y_4, y_5$ gives
\begin{align}
 z_1 = x_1\,,\quad
 z_2 = x_2\,,\quad
 z_3 = x_1 x_3\,,\quad
 z_4 = x_4\,,\quad
 z_5 = x_1 x_3 x_4 x_2 x_5 = x_1 x_2 x_3 x_4 x_5\,.
\end{align}
The algebra then decomposes as $[z_1,z_2] \ot [z_3,z_4] \ot [z_5]$,
with $z_5$ the central element.

To obtain a representation realizing the anti-commutativity relations we
make use of Eq.~\eqref{eq:reps} and identify,
\begin{align}
z_1 = XI\,, \quad
z_2 = ZI\,, \quad
z_3 = IX\,, \quad
z_4 = IZ\,, \quad
z_5 = \pm iII\,.
\end{align}
We chose the last element such that it satisfies $z_5^2 = (x_1 x_3 x_4 x_2 x_5)^2 = -1$.
Because $z_5$ is a central element, it must act as a scalar on each irreducible component, and thus $z_5 = \pm i II$.
Inverting the relations yields a realization of the anti-commutativity relations of the pentagon as Pauli strings,
\begin{align}\label{eq:penta_to_qubit}
x_1 &= z_1 = XI\,, \quad\quad
x_2 = z_2 = ZI\,, \quad\quad
x_3 = z_1 z_3            = XX\,, \nn\\
x_4 &= z_4 = IZ\,, \quad\quad
x_5 = z_4 z_3 z_2 z_5 = \pm ZY\,.
\end{align}
This is illustrated in Fig.~\ref{fig:pentagon}.
\end{example}

Naturally, Eq.~\eqref{eq:penta_to_qubit} is not a unique assignment;
in general other quasi-Clifford to Pauli mappings are possible by performing
a different sequences of splits.
But from the example one sees another feature:
if one does not require a mapping to Pauli strings that are indepedent (as generators of a group),
then one gets away with a smaller representation than given by~Theorem~\ref{thm:wedder},
namely it is enough to only consider the $[\gamma,\delta]$-algebras.

In contrast, if one requires the generators to be independent after the mapping,
then every $[\beta]$-algebra can be represented by a Pauli $Z$ acting on an additional tensor factor.
In this case,
\begin{align}
x_1 &= z_1 = IXI\,, \quad\quad
x_2 = z_2 = IZI\,, \quad\quad
x_3 = z_1 z_3            = IXX\,, \nn\\
x_4 &= z_4 = IIZ\,, \quad\quad
x_5 = z_4 z_3 z_2 z_5 = ZZY\,.
\end{align}
preserves the linear independence of the generators. The two matrix blocks induced by the additional Pauli $Z$
corresponds to the two signs in Eq.~\eqref{eq:penta_to_qubit}.

\begin{corollary}[Quasi-Clifford to qubit mapping]
 A quasi-Clifford to qubit mapping can be obtained by
 applying the splitting algorithm in Observation~\ref{prop:split}.
Then identify each $[\gamma, \delta]$-algebra with
the one-qubit Pauli algebra $[X, Z]$.
To preserve independence of the generators, identify every $[\beta]$-algebra by $[Z]$;
identify $[\beta]$ by $[1]$ if independence does not need to be preserved.
A full quasi-Clifford to qubit mapping is given by the tensor product representation in Observation~\ref{thm:wedder}, inverting the relations to the original variables $x_1, \dots, x_n$, and removing any complex phases $\pm i$ so to make all variables hermitian.
\end{corollary}

\section{Pauli to Pauli mappings}
More generally, this can also be used for a structural decomposition of Pauli groups,
providing a {\em Pauli-to-Pauli} mapping.
The key idea is to perform the splitting algorithm
but using the generators of the given Pauli group.
After the algorithm finishes,
the $j$'th connected pair of vertices is, up to a complex phase,
identified with logical
$X_j$ and $Z_j$ operators.
Isolated vertices are identified, up to a complex phase,
with a logical $Z_\ell$ on a free coordinate,
with the sign determined by the group relations.
This way, the resulting algebra satisfies the original group relations,
but is presented in a Wedderburn decomposition of the underlying QCA.

\begin{corollary}[Splitting algorithm for Pauli groups; see also~\cite{PhysRevA.79.062322} for quantum codes]\label{cor:split_pauli}
Let $G = \langle x_1, \dots, x_m \rangle \in \PP_m$ form a Pauli group.
Then the algorithm in Observation~\ref{prop:split} provides a Wedderburn decomposition for $G$.
The center is formed by
by classical binary systems $[Z_\ell]$ with $\ell = 1, \dots, r$,
while each pair of connected vertices forms a one-qubit Pauli algebra $[X_j,Z_j]$ with $j=r+1,\dots,r+s$.
\end{corollary}

\begin{proof}
 Note that at every iteration,
 the $y_1, \dots, y_m$ generate the same group as $x_1, \dots, x_m$.
 Let $z_1, \dots, z_m$ be the generators at the end of the algorithm.
 Then, by Observation~\ref{prop:split}, every isolated vertex forms a $[\beta]$-algebra,
 while every pair of vertices connected by an edge forms a $[\gamma, \delta]$-algebra.
 Over $\C$ these can be represented by single-qubit Pauli algebras [c.f. Eq.~\eqref{eq:reps}].
 It suffices to fix
 each isolated vertex $[\beta]$ and
 each pair of vertices $[\gamma, \delta]$
 separately:
Suppose $z_u$ forms a $[\beta]$-algebra with $z_u^2 = k_u$,
and $[z_v,z_w]$ forms a $[\gamma, \delta]$-algebra with
$z_w^2 = k_w$,  where $k_u,k_v,k_w \in \{-1,+1\}$.
Then identify the monomial $z_u$ by
\begin{align} R:
 z_u \quad &\mapsto \quad i^{{(1-k_u)}/2} Z\,.
 \end{align}
Identify the monomials $z_v$ and $z_w$ by
 \begin{align} S :
 z_v \quad &\mapsto \quad i^{{(1-k_v)}/2} X\,, \nn\\
 z_w \quad &\mapsto \quad i^{{(1-k_w)}/2} Z\,.
 \end{align}
Then $R$ and $S$ form faithful group representations
for the subgroups $\langle z_u \rangle$  and $\langle z_v, z_w\rangle$ respectively,
A representation of the Pauli group $\langle x_1, \dots, x_m\rangle$
is then given by inverting the relations for $z_1, \dots, z_m$
and the tensor product representation
$\phi = \bigotimes_{i=1}^r R_i \bigotimes_{j=r+1}^{r+s} S_j$.
\end{proof}

\section{Recovering Jordan-Wigner}
In order to set quasi-Clifford to qubit mappings into context,
let us consider the Jordan-Wigner transformation.
A fermionic system in second quantization is described by
$N$ creation $\{a_i^\dag\}_{0=1}^{N-1}$ and $N$ annihilation $\{a_i\}_{i=0}^{N-1}$ operators.
These satisfy $\{a_i, a_j\} = \{a_i^\dag, a_j^\dag\} = 0$ and $\{a_i^\dag, a_j\} = \delta_{ij}\one$.
A fermionic observable is then a linear combination of square-free monomials in the $a_j^\dag$ and $a_j$.

Another useful basis for fermionic systems is that of $2N$ Majorana operators,
\begin{align}
m_{2j} &= \phantom{i}(a_j + a_j^\dag)\,,\quad\quad
m_{2j+1} = i(a_j^\dag - a_j)\,.
\end{align}
These satisfy the anti-commutation relations
$\{m_k, m_\ell\} = 2\delta_{k\ell}\one$,
so that the corresponding anti-commutativity graph is fully connected.
The Jordan-Wigner transformation establishes the following
mapping of Majorana operators to Pauli strings~\footnote{
The convention is chosen so to coincide with $X$ and $Z$ as single-qubit generators.},
\begin{align}\label{eq:majorana}
 m_{2j}   & \quad\mapsto\quad \tilde m_{2j} = X_{j} \prod_{k=0}^{j-1} Y_k\,,  \quad\quad\quad
 m_{2j+1} \quad\mapsto\quad \tilde m_{2j+1} = Z_{j} \prod_{k=0}^{j-1} Y_k\,,
\end{align}
for $j=0,\dots, N-1$.
In particular, the $\tilde m_i$ have the same anti-commutativity relations as the $m_i$,
that is, they all mutually anti-commute. In addition,
the $\tilde m_i$ have acquired a group structure originating from the Pauli relations.

The Jordan-Wigner transformation
is a special case of the algorithm in Cor.~\ref{cor:split_pauli}.

\begin{proposition}
Corollary~\ref{cor:split_pauli} gives the Jordan-Wigner transformation
for an anti-commuting set of even size.
\end{proposition}
\begin{proof}
Let $A = \{x_1, \dots, x_m\}$ a set of pairwise anti-commutating hermitian unitary operators.
Then its corresponding anti-commutativity graph is fully connected.
For every pair of vertices that is split from the graph, according to Eq.~\eqref{eq:stay}
the remaining graph remains fully connected with its vertices transformed
by $y_i \mapsto y_k y_\ell y_i$ where $y_k, y_\ell$ is the pair that is being split.
At the end, only pairs of vertices are left, in addition a single vertex if $|A|$ is odd.
The Wedderburn reduced basis reads
\begin{align}\label{eq:rec_JW}
z_1 &= x_1\,,\nonumber\\
z_2 &= x_2\,,\nonumber\\
z_i &= \begin{cases}
              \big(\prod_{k=1}^{i-1} x_k \big) x_i \quad\quad \text{$i$ odd}\,,  \\
              \big(\prod_{k=1}^{i-2} x_k \big) x_i  \quad\quad \text{$i$ even}\,.
           \end{cases}
 \end{align}
 Inverting the relations yields, up to a sign,
\begin{align}\label{eq:rec_JW_inv}
 x_1 &= z_1\,,\nonumber \\
 x_2 &= z_2\,,\nonumber \\
 x_i &= \begin{cases}
              \big(\prod_{k=1}^{i-1} z_k \big) z_i \quad\quad \text{$i$ odd}\,,  \\
              \big(\prod_{k=1}^{i-2} z_k \big) z_i \quad\quad \text{$i$ even}\,.
           \end{cases}
 \end{align}
Mapping to qubits with $x_{2j} \mapsto X_j$, $x_{2j+1} \mapsto Z_j$ gives
the Jordan-Wigner transformation [Eq.~\eqref{eq:majorana}].
\end{proof}

\section{Symmetry-reducing semidefinite programs}
Many semidefinite programs in quantum information deal with problems on multi-qubit systems, and often these SDPs are invariant under some set of operations. Because the mapping of Corollary~\ref{cor:split_pauli} preserves the ${}^\dag$-operation, it also preserves the property of a matrix being positive-semidefinite.
This allows to symmetry-reduce positive semidefinite operators~\cite{Bachoc2012}.

To see how this works,
consider finding the ground state of a Hamiltonian defined on a graph $G =(V,E)$,
\begin{equation}\label{eq:Hamiltonian}
H = - \sum_{\{i,j\} \in E} (XZ + ZX)_{ij}\,.
\end{equation}
A ground state can be found by solving the following semidefinite program,
\begin{align}\label{eq:Hamiltonian_SDP}
 \text{minimize}_{\varrho} \quad  &\tr(H \varrho)\quad\quad  \text{subject to} \quad\quad \varrho \succeq 0 \quad\text{and} \quad \tr(\varrho) = 1\,.
\end{align}
One way to approximate the ground state of large systems is to
relax the optimization of Eq.~\eqref{eq:Hamiltonian_SDP}
to that over locally consistent marginals,
\begin{align}\label{eq:Hamiltonian_relax}
 \text{minimize} \quad  &\tr(H \varrho_{ijk})           \nn\\
 \text{subject to} \quad& \varrho_{ijk} \succeq 0  \nn\\
                   & \tr(\varrho_{ijk}) = 1   \nn\\
                  & \tr_k(\varrho_{ijk}) = \tr_\ell (\varrho_{ij\ell})\,,
\end{align}
where the optimization is over all distinct $i,j,k,\ell \in \{1,\dots, n\}$.

The optimization in Eq.~\eqref{eq:Hamiltonian_relax} is over a set of locally consistent three-qubit marginals, and thus over $\binom{n}{3}$ complex matrices of size $8\times 8$.
The smallest algebra containing the Pauli terms of a single three-qubit reduced Hamiltonian is
\begin{equation}
 S = [XZI, ZXI, XIZ, ZIX, IXZ, IZX]\,.
\end{equation}
Note that the last interaction term $IZX$ is not independent.
Using Corollary~\ref{cor:split_pauli} one can find a $\star$-isomorphism:
\begin{align}
 \varphi : XZI \quad&\mapsto\quad IXI\,,
& XIZ \quad&\mapsto\quad IIZ\,,
& IXZ \quad&\mapsto\quad \phantom{-}ZXZ \nn\\
 ZXI \quad&\mapsto\quad IIX\,,
& ZIX \quad&\mapsto\quad IZI\,,
& IZX \quad&\mapsto\quad -ZXZ\,.
\end{align}
Because this mapping preserves the ${}^\dag$-operation,
it also preserves the property of a matrix being positive-semidefinite~\cite{Bachoc2012}.
Now note that after applying $\varphi$, the first coordinate contains only Pauli I and Z terms, yielding a block-diagonalization of $\varrho_{ijk}$ into two $4 \times 4$ matrix blocks.

\section{Maximal anti-commuting sets}
Given a Pauli group, how can one find a maximal subset in which every pair of operators anti-commute?
We recall that any maximal anti-commuting subset must be of odd size:
\begin{lemma}
 Let $G$ be a Pauli group.
 If $A$ is a maximal set of anti-commuting Pauli elements in $G$, then $|A|$ is odd.
\end{lemma}
\begin{proof}
Let $A = \{x_1, \dots, x_m\}$.
Suppose $a$ is a maximal anti-commuting set with $|A|$.
Then $x_1 \cdots x_m \in G$ anti-commutes with all $x_i$
and an additional element of $G$ can be added to $A$,
yielding a contradiction.
Thus $|A|$ must be odd.
\end{proof}

In particular, given a anti-commuting set $\{ g_1, \dots, g_m \}$ with $m$ even,
it can be checked that the only element in the group $\langle g_1, \dots, g_m\rangle$
that anti-commutes with all of them is the monomial $g_1 g_2\cdots g_m$.
If $m$ is odd, no additional element can be found.

\begin{observation}\label{eq:construct_maxacommset}
Given a Pauli group, the following algorithm constructs a maximally anti-commuting subset.
 \begin{enumerate}
 \item
 Use the splitting algorithm to decompose the group into a tensor product of $[\beta]$- and $[\gamma,\delta]$-algebras,
 \begin{equation*}
  [g_1, \dots ,g_m ] \simeq  [\beta_1] \ot \dots \ot [\beta_r] \ot \dots \ot [\gamma_1,\delta_1] \ot \dots \ot [\gamma_s, \delta_s]\,,
 \end{equation*}
 where $r+2s = m$.
 \item Map every $[\gamma,\delta]$-algebra to a $[X,Z]$-algebra.
 \item To the pairs of vertices apply the inverse Jordan-Wigner transformation [Eq.~\eqref{eq:rec_JW}], obtaining an anti-commuting set $\{x_1, \dots,x_s, x_{s+1}, \dots, x_{2s}\}$.

 \item Complete the anti-commuting set with the only monomial that anti-commutes with all $x_i$,
 that is with $x_{2s+1} = x_1 \cdots x_s x_{s+1} \cdots x_{2s}$.
 The resulting maximal anticommuting set has size $2s + 1$.
 \end{enumerate}

\end{observation}
\begin{proof}
The elements of the set anti-commute by construction.
To see that the set is maximal, note that the algebra generated by any anti-commuting subset of size $2s+1$
is isomorphic to $[\beta] \ot [\gamma_1, \delta_1] \ot \dots \ot [\gamma_s, \delta_s] $,
which is realized by the inverse Jordan-Wigner transformation [Eq.~\eqref{eq:rec_JW}].
Thus if $G=[g_1, \dots, g_m]$ decomposes into $s$ $[\gamma,\delta]$ algebras, the largest anti-commuting subset of $G$ has size $2s+1$, which is realized by Observation~\ref{eq:construct_maxacommset}.
\end{proof}

\begin{example}
 Consider the Pauli group $G = \langle XXI, XIX, ZZI, ZIZ\rangle$.
 Corollary~\ref{cor:split_pauli} yields the $\ast$-isomorphism
 \begin{equation}\label{eq:map_to_real2qb}
\phi: \quad
III \mapsto  \phantom{-} II\,, \quad
XXI \mapsto  \phantom{-} XI\,,  \quad XIX \mapsto \phantom{-} IX\,, \quad
ZZI \mapsto  \phantom{-} IZ\,,  \quad ZIZ \mapsto \phantom{-} ZI\,.
\end{equation}
Following the procedure in Observation~\ref{eq:construct_maxacommset} we apply the inverse Jordan-Wigner transformation. Ignoring phases, a maximal anti-commuting subset of $G$ is
 $\{x_1 = XXI\,,
 x_2 = ZIZ\,,
 x_3 = ZXY\,,
 x_4 = XYZ \,,
 x_5 = IYY\}$,
where $x_5 = x_1x_2x_3x_4$.
\end{example}

\bibliographystyle{amsalpha}
\bibliography{current_bib}

\appendix

\section{Relations after splitting}
\label{app:rel_after_split}

More precisely, the relations after splitting $(x_1,x_2)$ are
(see Fig.~\ref{fig:split_situations} A, B):
the relations remain,
\begin{align}\label{eq:stay}
 y_i \sim y_j \text{\quad iff \quad} x_i \sim x_j\quad
 \begin{cases}
& \text{if } x_i \not\sim x_1 \text{\quad and\quad} x_i \not\sim x_2 \\
& \text{if } x_j \not\sim x_1 \text{\quad and\quad} x_j\not\sim x_2\\
& \text{if } x_i \sim x_1 \text{\quad and\quad} x_j \sim x_1
  \text{\quad and\quad} x_i \sim x_2 \text{\quad and\quad} x_j \sim x_2
  \end{cases}
\end{align}
The relations invert
(see Fig.~\ref{fig:split_situations} C, D):
\begin{align}\label{eq:invert}
 y_i \sim y_j \text{\quad iff\quad} x_i \not\sim x_j\quad
 \begin{cases}
& \text{if }
x_i \not\sim x_1 \text{\quad and\quad}
x_i     \sim x_2 \text{\quad and\quad}
x_j     \sim x_1 \text{\quad and\quad}
x_j\sim x_2 \\
& \text{if }
x_i     \sim x_1 \text{\quad and\quad}
x_i     \sim x_2 \text{\quad and\quad}
x_j \not\sim x_1 \text{\quad and\quad}
x_j\sim x_2 \\
& \text{if }
x_i     \sim x_1 \text{\quad and\quad}
x_i \not\sim x_2 \text{\quad and\quad}
x_j     \sim x_1 \text{\quad and\quad}
x_j     \sim x_2 \\
& \text{if }
x_i     \sim x_1 \text{\quad and\quad}
x_i     \sim x_2 \text{\quad and\quad}
x_j     \sim x_1 \text{\quad and\quad}
x_j \not\sim x_2 \\
& \text{if }
x_i     \sim x_1 \text{\quad and\quad}
x_i \not\sim x_2 \text{\quad and\quad}
x_j \not\sim x_1 \text{\quad and\quad}
x_j     \sim x_2 \\
& \text{if }
x_i \not\sim x_1 \text{\quad and\quad}
x_i     \sim x_2 \text{\quad and\quad}
x_j     \sim x_1 \text{\quad and\quad}
x_j \not\sim x_2
\end{cases}
\end{align}

\end{document}